\documentclass[12pt]{amsart}
\usepackage{txfonts}
\usepackage{amssymb}
\usepackage{eucal}
\usepackage{graphicx}
\usepackage{amsmath}
\usepackage{amscd}
\usepackage[all]{xy}           

\usepackage{amsfonts,latexsym}
\usepackage{xspace}
\usepackage{epsfig}
\usepackage{float}
\usepackage{color}
\usepackage{fancybox}
\usepackage{colordvi}
\usepackage{multicol}
\usepackage{colordvi}
\usepackage{stmaryrd}
\usepackage[colorlinks,final,backref=page,hyperindex,hypertex]{hyperref}

\newtheorem{theorem}{Theorem}[section]
\newtheorem{defn}[theorem]{Definition}
\newtheorem{lemma}[theorem]{Lemma}
\newtheorem{coro}[theorem]{Corollary}
\newtheorem{prop-def}{Proposition-Definition}[section]

\newcommand{\nc}{\newcommand}
\newcommand{\delete}[1]{}

\nc{\mlabel}[1]{\label{#1}}  
\nc{\mcite}[1]{\cite{#1}}  
\nc{\mref}[1]{\ref{#1}}  
\nc{\mbibitem}[1]{\bibitem{#1}} 

\delete{
\nc{\mlabel}[1]{\label{#1}  
{\hfill \hspace{1cm}{\bf{{\ }\hfill(#1)}}}}
\nc{\mcite}[1]{\cite{#1}{{\bf{{\ }(#1)}}}}  
\nc{\mref}[1]{\ref{#1}{{\bf{{\ }(#1)}}}}  
\nc{\mbibitem}[1]{\bibitem[\bf #1]{#1}} 
}

\nc{\bfk}{\mathbf{k}}
\nc{\Der}{\mathrm{Der}}
\nc{\Ker}{\mathrm{Ker}}

\begin{document}

\title{Homogeneous Rota-Baxter operators on $A_{\omega}$ (II) }

\author{RuiPu  Bai}
\address{College of Mathematics and Information  Science,
Hebei University, Baoding 071002, China} \email{bairuipu@hbu.edu.cn}

\author{Yinghua Zhang}
\address{College of Mathematics and Information  Science,
Hebei University, Baoding 071002, China} \email{zhangyinghua1234@163.com}

\date{\today}

\begin{abstract} In this paper we study  $k$-order homogeneous  Rota-Baxter operators with weight $1$ on the  simple $3$-Lie algebra $A_{\omega}$ (over a field of characteristic zero),  which is realized by an associative commutative algebra
$A$ and  a derivation $\Delta$ and an involution $\omega$ (Lemma
\mref{lem:rbd3}). A $k$-order homogeneous  Rota-Baxter operator on
$A_{\omega}$ is a linear map $R$  satisfying
$R(L_m)=f(m+k)L_{m+k}$ for all generators $\{ L_m~ |~ m\in \mathbb Z \}$ of $A_{\omega}$ and  a map $f :
\mathbb Z \rightarrow\mathbb F$, where  $k\in \mathbb Z$. We prove that $R$ is a $k$-order  homogeneous
Rota-Baxter operator on $A_{\omega}$ of weight $1$ with $k\neq 0$ if and only if $R=0$ (see Theorems \ref{thm:3.1}), and $R$ is a $0$-order  homogeneous
Rota-Baxter operator on $A_{\omega}$ of weight $1$ if and only if $R$ is  one
of the forty possibilities which are described
in  Theorems \ref{thm:fin},  \ref{thm:Rm0},  \ref{thm:Rm1},  \ref{thm:R01},   \ref{thm:f(0)a},   \ref{thm:f(0)a1} and  \ref{thm:f(0)a3}.

\end{abstract}

\subjclass[2010]{17B05, 17D99.}

\keywords{ $3$-Lie algebra, homogeneous  Rota-Baxter operator, Rota-Baxter $3$-algebra.}

\maketitle



\allowdisplaybreaks

\section{Introduction}

Rota-Baxter operators  have been closely related to many fields in mathematics and mathematical physics.
They have played an important role in the Hopf algebra approach of renormalization of perturbative quantum field theory \mcite{ BGN, BGN2, EGK,EGM}, as well as in the application of the renormalization method in solving divergent problems in number theory~\mcite{GZ,MP}, they are also important topics
in many fields such as symplectic geometry, integrable systems, quantum groups and
quantum field theory~\cite{Ag2,BBGN,Ca,EGK,Guw,Gub,GK1, GK3, HuiB, GZ,Ro1,Ro2}.

Authors in  ~\cite{BGL11} investigated the Rota-Baxter operators
on $n$-Lie algebras  \cite{F} and  studied  the structure of Rota-Baxter
$3$-Lie algebras,  and they also provided a method to realize Rota-Baxter
$3$-Lie algebras from Rota-Baxter $3$-Lie algebras, Rota-Baxter
Lie algebras, Rota-Baxter pre-Lie algebras and Rota-Baxter
commutative associative algebras and derivations. In paper ~\cite{BZ}, authors discussed a class of Rota-Baxter operators of weight zero  on an
infinite dimensional simple 3-Lie algebra $A_{\omega}$ over a field $\mathbb F$ of characteristic zero, which is  the $0$-order homogeneous Rota-Baxter operators of weight zero.  A  homogeneous  Rota-Baxter operator on
$A_{\omega}$ is a linear map $R$  satisfying
$R(L_m)=f(m)L_{m}$ for all generators $\{ L_m~ |~ m\in \mathbb Z \}$ of $A_{\omega}$ and  a map $f :
\mathbb Z \rightarrow\mathbb F$.
It is proved that $R$ is a homogeneous
Rota-Baxter operator on $A_{\omega}$ if and only if $R$ is one of the five possibilities $R_{0_1}, R_{0_2}, R_{0_3}, R_{0_4}$ and
$R_{0_5}$. By means of homogeneous Rota-Baxter
operators, new 3-Lie algebras $(A,[,,]_i)$  for $1\leq i \leq 5$ are constructed,  and $R_{0_i}$ is also  an
homogeneous Rota-Baxter operator on the 3-Lie algebra $(A,[,,]_i$), for $1\leq i\leq 5$, respectively.

In this paper we investigate  $k$-order homogeneous Rota-Baxter
operators of weight $1$ on the simple $3$-Lie $\mathbb F$-algebra $A_{\omega}$, where $\mathbb F$ is a field of characteristic zero. Throughout this
paper, by an algebra we mean an $\mathbb F$-algebra and we denote by $\mathbb Z$ the set of integers.

\section{preliminary}
\label{sec:rbn}

We recall that a {\bf 3-Lie algebra} over a field $\mathbb F$ is an $\mathbb F$-vector space $A$  endowed with a ternary multi-linear skew-symmetric operation
satisfying
for all $x_1,x_2,x_3, y_2, y_3\in A$.
\begin{equation}\label{eq:2.1}
[[x_1,x_2,x_3],y_2,y_3]=[[x_1,y_2,y_3],x_2,x_3] +[[x_2,y_2,y_3],x_3,x_1]+[[x_3,y_2,y_3],x_1,x_2].
\end{equation}

\begin{defn}
\mlabel{de:nlie}  Let $\lambda\in\mathbb F$ be fixed.
A  {\bf Rota-Baxter $3$-algebra} is a $3$-algebra $(A,\langle , , \rangle)$ with a linear map $R: A\to A$ such that
\begin{eqnarray}\label{eq:2.3}
\langle R(x_1), R(x_2), R(x_3)\rangle
&=& R\Big(
\langle R(x_1), R(x_2),x_3\rangle +\langle R(x_1),x_2, R(x_3)\rangle +\langle x_1, R(x_2), R(x_3)\rangle \notag \\
&&
+\lambda \langle R(x_1),x_2,x_3\rangle
+\lambda \langle x_1, R(x_2),x_3\rangle
+\lambda \langle x_1,x_2, R(x_3)\rangle
\mlabel{eq:rb3de}\\
&&
+\lambda^2 \langle x_1,x_2,x_3\rangle\Big).
\notag
\end{eqnarray}
\label{de:2.1}
\end{defn}

\begin{lemma}

Let $(A, \langle\ , , \rangle)$ be a  $3$-algebra over $\mathbb F$, $R: A\rightarrow A$ be a linear map and $\lambda\in \mathbb F$, $\lambda\neq 0$.
Then $(A, \langle\ , , \rangle, R)$ is a Rota-Baxter $3$-algebra of weight $ \lambda$  if and only if   $(A, \langle\ , , \rangle, \frac{1}{\lambda}R)$ is a Rota-Baxter $3$-algebra of weight $1$.
\label{lem:rbd2}
\end{lemma}

\begin{proof} Apply Eq
(\ref{eq:2.3}).
\end{proof}

\begin{lemma}\cite{BW22} Let $A$ be an $\mathbb F$-vector space with a basis $\{ L_n~|~ n\in \mathbb Z \}$. Then $A$ is a simple $3$-Lie algebra
 in the multiplication
\begin{equation}
{[}L_{l},L_{m},L_{n}]=\begin{vmatrix}
(-1)^l&(-1)^m&
   (-1)^n \\
1&1&1  \\
l& m&
   n \\
\end{vmatrix}L_{l+m+n-1}, ~~\mbox{for all } ~~l, m, n\in \mathbb Z.
\mlabel{eq:defthlmn}
\end{equation}
\mlabel{lem:rbd3}
\end{lemma}

\noindent{\bf Notation}. In the following, the 3-Lie algebra $A$ in Lemma \mref{lem:rbd3} is denoted by ${\bf A_{\omega}}$, and we
set

\begin{equation}\label{eq:det}
D(l, m, n) := \begin{vmatrix}
(-1)^l&(-1)^m&
   (-1)^n \\
1&1&1  \\
l& m&
   n \\
\end{vmatrix}.
\end{equation}

\begin{lemma}\cite{BZ}\quad
 $D(l, m, n)=0$  if and only if  for all $l, m, n, k, s, t\in \mathbb Z,$

 \vspace{2mm} $(l-m)(l-n)(m-n)=0$, or $l=2k+1, m=2s+1, n=2t+1$,  or $l=2k, m=2s, n=2t$.
\label{lem:det}
 \end{lemma}

\section{Homogeneous Rota-Baxter operators of weight $1$ on $3$-Lie algebra $A_{\omega}$}
\mlabel{sec:rbl3rbl}

By Definition~\ref{de:2.1}, if $(A, [, , ], R)$ is a Rota-Baxter $3$-Lie
algebra of weight $1$, then the $\mathbb F$-linear map $R: A\rightarrow
A$ satisfies, for all $x_1, x_2, x_3\in A$,
\begin{eqnarray}\label{eq:3.1}
[ R(x_1), R(x_2), R(x_3)]
&=& R\big(
[R(x_1), R(x_2),x_3] +[ R(x_1),x_2, R(x_3)] + [ x_1, R(x_2), R(x_3)] \notag \\
&&
+ [ R(x_1),x_2,x_3]
+ [ x_1, R(x_2),x_3 ]
+ [ x_1,x_2, R(x_3)]\\
&&
+[ x_1,x_2,x_3]\big).
\notag
\end{eqnarray}

\begin{defn}
\mlabel{de:krotabaxter}

Let $R$ be a   Rota-Baxter operator on the $3$-Lie algebra $A_{\omega}$. If there exist a map
$f: \mathbb Z \rightarrow \mathbb F$, and $k\in \mathbb Z$ such that
\begin{equation}\label{eq:3.2}
R(L_{m})=f(m+k)L_{m+k},~~  \forall m\in \mathbb  Z,
\end{equation}
then $R$ is called {\bf a $k$-order  homogeneous Rota-Baxter operator}, which is denoted by $R_k$.

\end{defn}

\subsection{$k$-order homogeneous Rota-Baxter operators with $k\neq 0$}

From Eq \eqref{eq:3.2}, we know that for all ~$x,y,z\in A_\omega$,

\vspace{2mm} $
[R_{k}(L_{l}),R_{k}(L_{m}),R_{k}(L_{n})]=[f(l+k)L_{l+k},f(m+k)L_{m+k},f(n+k)L_{n+k}]
$
\\
$
=f(l+k)f(m+k)f(n+k)D(l+k, m+k, n+k)L_{l+m+n+3k-1},
$

\vspace{2mm} $
R_{k}\Big([L_{l},R_{k}(L_{m}),R_{k}(L_{n})]+[R_{k}(L_{l}),L_{m},R_{k}(L_{n})]+[R_{k}(L_{l}),R_{k}(L_{m}),L_{n}]
$
\\
$
+[R_{k}(L_{l}),L_{m},L_{n}]+[L_{l},R_{k}(L_{m}),L_{n}]+[L_{l},L_{m},R_{k}(L_{n})]+[L_{l},L_{m},L_{n}]\Big)
$
\\
$
=R_{k}([L_{l},f(m+k)L_{m+k},f(n+k)L_{n+k}]+[f(l+k)L_{l+k},L_{m},f(n+k)L_{n+k}]
$
\\
$
+[f(l+k)L_{l+k},f(m+k)L_{m+k},L_{n}])+[f(l+k)L_{l+k},L_{m},L_{n}]
$
\\
$
+[L_{l},f(m+k)L_{m+k},L_{n}]+[L_{l},L_{m},f(n+k)L_{n+k}]+[L_{l},L_{m},L_{n}])
$
\\
$
=f(m+k)f(n+k)f(l+m+n+3k-1)D(l, m+k, n+k)L_{l+m+n+3k-1}
$
\\
$
+f(l+k)f(n+k)f(l+m+n+3k-1)D(l+k, m, n+k)L_{l+m+n+3k-1}
$
\\
$
+f(l+k)f(m+k)f(l+m+n+3k-1)D(l+k, m+k, n)L_{l+m+n+3k-1}
$
\\
$
+f(l+k)f(l+m+n+2k-1)D(l+k, m, n)L_{l+m+n+2k-1}
$
\\
$
+f(m+k)f(l+m+n+2k-1)D(l, m+k, n)L_{l+m+n+2k-1}
$
\\
$
+f(n+k)f(l+m+n+2k-1)D(l, m, n+k)L_{l+m+n+2k-1}
$
\\
$
+f(l+m+n+k-1)D(l, m, n)L_{l+m+n+k-1}.
$

\vspace{2mm} Thanks to Eq \eqref{eq:3.1},

\vspace{2mm} $
[f(l+k)L_{l+k},f(m+k)L_{m+k},f(n+k)L_{n+k}]
$
\\
$=R_{k}([L_{l},f(m+k)L_{m+k},f(n+k)L_{n+k}]+[f(l+k)L_{l+k},L_{m},f(n+k)L_{n+k}]
$
\\
$+[f(l+k)L_{l+k},f(m+k)L_{m+k},L_{n}]).$

\vspace{2mm} Therefore,  if $k\neq 0$, then   for all~$l, m, n\in \mathbb  Z$, $R_{k}([L_{l},L_{m},L_{n}])=0$.
Thanks to  $A_\omega=[A_\omega, A_\omega, A_\omega]$,  $R_k(A_\omega)=0$.

This shows the following result.

\vspace{2mm}\begin{theorem} A linear map $R_{k}$ defined by Eq \eqref{eq:3.2} is a $k$-order homogeneous Rota-Baxter operator of weight $1$  on $A_\omega$
if and only if $R_k=0$.
\label{thm:3.1}
\end{theorem}

\subsection{$0$-order homogeneous Rota-Baxter operators of weight $1$}

In the following we discuss the $0$-order homogeneous Rota-Baxter operators of weight $1$ on $A_{\omega}$. Then   Eq \eqref{eq:3.2} is reduced to

\begin{equation}
R(L_{m})=f(m)L_{m}, \forall m\in \mathbb Z.
\label{eq:R0}
\end {equation}

{\it For convenience, throughout this paper we suppose that $R$ is a linear map on $A_{\omega}$ defined by Eq \eqref{eq:R0},  and
 $0$-order homogeneous Rota-Baxter operator  $R_0$ of weight $1$ on $A_{\omega}$ is simply denoted by $R$, and  is simply  called {\bf  a  homogeneous Rota-Baxter operator on $A_\omega$.}

 Denote

 $W_1=\{ 2m ~|~ m\in \mathbb Z, m\neq 0, f(2m)\neq 0 \}$, \quad $U_1=\{ 2m+1 |~ m\in \mathbb Z, m\neq 0, f(2m+1)\neq 0\}$,

  $W_2=\{ 2m ~ | ~m\in \mathbb  Z, m\neq 0,  f(2m)=0 ~ \}$, \hspace{2mm} $U_2=\{ 2m+1 | ~m\in \mathbb Z, m\neq 0,  f(2m+1)=0\}$.}

\begin{lemma}  The linear map $R$ is a  homogeneous Rota-Baxter operator on $A_\omega$
if and only if  the map $f: \mathbb Z\rightarrow \mathbb F $ in Eq \eqref{eq:R0} satisfies that for all $l, m, n\in \mathbb Z$,
\begin{equation}
\hspace{-2cm} f(2l+1)f(2m+1)f(2n)=(f(2l+1)f(2m+1)+f(2l+1)f(2n)
\label{eq:odd}
\end{equation}

\hspace{1cm}$+f(2m+1)f(2n)
+f(2l+1)+f(2m+1)+f(2n)+1)f(2l+2m+2n+1), l\neq m.$

\begin{equation}
\hspace{-1cm}f(2l+1)f(2m)f(2n)=(f(2l+1)f(2m)+f(2l+1)f(2n)+f(2m)f(2n)
\label{eq:even}
\end{equation}

\hspace{1cm}
$+f(2l+1)+f(2m)+f(2n)+1)f(2l+2m+2n), m\neq n.$
\label{lem:3.1}
\end{lemma}

\begin{proof} By Eq \eqref{eq:3.1} and Eq \eqref{eq:R0}, $R$ is a  homogeneous Rota-Baxter operator on $A_\omega$
if and only if  $f$ satisfies that for all $l, m, n\in \mathbb Z$,

$
f(l)f(m)f(n)D(l, m, n)
=\Big(f(l)f(m)+f(l)f(n)+f(m)f(n)+f(l)+f(m)
$

\hspace{4cm}$
+f(n)+1\Big)
f(l+m+n-1)D(l, m, n).
$
\\
Follows from Lemma \ref{lem:det}, we obtain the result.
\end{proof}

 From Eq \eqref{eq:odd} and Eq \eqref{eq:even}, for $l=n=0$, and $m\in \mathbb Z, m\neq 0, 1$, we have
 $$f(0)f(m)f(1)=(f(0)f(1)+f(m)f(1)+f(0)f(m)+f(0)+f(1)+f(m)+1)f(m),$$
\\ so we get
\begin{equation}\label{eq:3.7}
(f(0)+f(1)+1)f(m)(f(m)+1)=0.
\end{equation}
Therefore, we will start the discussion according to  the value $f(0)+f(1)+1$.

\subsubsection{\bf Homogeneous Rota-Baxter operators with $f(0)+f(1)+1\neq 0$}

In this section we discuss homogeneous Rota-Baxter operators $R$ on $A_{\omega}$ defined by Eq \eqref{eq:R0} of the case  $f(0)+f(1)+1\neq 0$.

\begin{lemma} Let  $R$ be a   homogeneous Rota-Baxter operator on $A_{\omega}$. Then the map $f: \mathbb Z\rightarrow \mathbb F$ in Eq \eqref{eq:R0} satisfies equation
\begin{equation}\label{eq:fm}
f(m)(f(m)+1)=0, ~~ \forall m\in \mathbb Z, m\neq 0, 1.
\end{equation}

\label{lem:fm}
\end{lemma}

\begin{proof} The result follows from $f(0)+f(1)+1\neq 0$, and Eq \eqref{eq:3.7}, directly.
\end{proof}

\begin{theorem} If at least one of the subsets $W_i, U_i$, $i=1, 2$ is finite. Then $R$ is a  homogeneous Rota-Baxter operator on $A_{\omega}$ if and only if the map $f: \mathbb Z\rightarrow \mathbb F$ in Eq \eqref{eq:R0} satisfies one  of the following,
for all $m, n\in \mathbb Z$,

 1) $f(m)=0$;

 2) $f(m)=-1$;

 3) $f(2m)=0, f(2m+1)=-1, m\neq 0$, and $f(0)(f(1)+1)=0$;

 4) $f(2m)=-1, f(2m+1)=0,  m\neq 0$ and  $f(1)(f(0)+1)=0$.

\label{thm:fin}
\end{theorem}

\begin{proof} If $f$ satisfies one of the cases 1) - 4). By a direct computation, we know that $R$ satisfies Eq \eqref{eq:odd} and Eq \eqref{eq:even}, that is, $R$
is a  homogeneous Rota-Baxter operator on $A_{\omega}$.

\vspace{2mm}Conversely, suppose that $R$ is a homogeneous Rota-Baxter operator on $A_{\omega}$.

First, we prove that if $W_i$ ( or $U_i$) is a finite subset, then $W_i$ ( or $U_i$) is empty, where $i=1$ or $ 2$.

 Without loss of generality, we can suppose that $|W_1| < \infty$.

If $|W_1|=s, $ and $1\leq s <\infty$. Suppose $W_1=\{ 2m_0, \cdots, 2m_{s-1}\}$, $s \geq 1$. Then $|W_2|=\infty$. Without loss of generality, we can suppose that $|U_1|\neq 0.$
Then there is $n_0\neq 0$ such that $f(2n_0+1)=-1$. We assert that  $|U_2| <\infty$ and $|U_1|=\infty.$

In fact, if $|U_2|=\infty$. Then there exist $2m, 2n\in W_2$, and $2l+1\in U_2$ such that  $m\neq n$ and $2m+2n+2l=2m_0$. By  Eq \eqref{eq:even},
we get the contradiction
$0=f(2m)f(2n)f(2l+1)=f(2m_0).$  Therefore, $|U_2| <\infty$, and $|U_1|=\infty$. So there exist $2l+1, 2n+1\in U_1$, and $2m\in W_2$ such that  $l\neq n$, and $2m+2n+2l=2n_0$.
We get the contradiction
$0=f(2m)f(2n+1)f(2l+1)=f(2n_0+1).$

Summarizing above discussion, we obtain that $W_1$ is empty, that is,
$f(2m)=0$ for all $m\in \mathbb Z, m\neq 0.$

Second we discuss the characteristic of $f$.

$\bullet$  If $U_2$ is non-empty, then there is   $2n_0+1\in U_2$ such that $f(2n_0+1)= 0$.  By Eq \eqref{eq:odd} and Eq \eqref{eq:even},  for all $m\neq -n_0$ and $m\neq  0$, $f$ satisfies that
$$f(2n_0+1)f(2m)f(-2n_0-2m)=f(0)=0,$$ ~~~
$$f(2n_0+1)f(1)f(-2n_0)=(f(1)+1)f(1)=0.$$
\\Thanks to $f(0)+f(1)+1\neq 0$,  $f(0)=f(1)=0$. Again by Eq \eqref{eq:odd}, for all $m\in \mathbb Z$,
$$f(2n_0+1)f(1)f(2m)=f(2n_0+2m+1)=0,$$
we obtain that for all $l\in \mathbb Z$, $l\neq -n_0$, $f(2l+1)=0$. By the similar discussion to the above, we obtain that for all $l\in \mathbb Z$, $f(2l+1)=0$. This is the case 1).

$\bullet\bullet$  If $U_2$ is empty, then  for all $l\in \mathbb Z, l\neq 0$, $f(2l+1)=-1$. Thanks to  Eq \eqref{eq:odd} and  Eq \eqref{eq:even},  $f(0)(f(1)+1)=0.$ This is the case 3).

$\bullet$$\bullet$$\bullet$ Similarly, if $W_2$  is empty, then  for all $m\in \mathbb Z, m\neq 0$, $f(2m)=-1$. By the similar discussion,  we obtain the cases 2) and 4).

If $U_1$  is empty, then for all $m\in \mathbb Z, m\neq 0$,  $f(2m+1)=0$. We obtain the cases 1) and 4).

 If $U_2$  is empty, then for all $m\in \mathbb Z, m\neq 0$,  $f(2m+1)=-1$.   We obtain the cases 2) and 3).
\end{proof}

Now we discuss the case  $|W_i|=|U_i|=\infty, $ for $i=1, 2.$

\begin{lemma} Let $R$ be a   homogeneous Rota-Baxter operator on $A_{\omega}$. If $W_1=\{ 2m_i |  m_i < m_{i+1}, i\in \mathbb Z, i\geq 0 \}$.
Then $U_1=\{ 2l_i+1 | l_i < l_{i+1}, i\in \mathbb Z, i\geq 0 \}$, and $l_0\geq -m_1$, $l_1\geq -m_0$.
\label{lem:Um0}
\end{lemma}

\begin{proof}
 For all $2l+1\in U_1$,  by Eq \eqref{eq:even}, we have $f(2m_0+2m_1+2l)=-1$. Then $2l+2m_0+2m_1\geq m_0$, we obtain $l\geq -m_1$. So we can suppose that
 $U_1=\{ 2l_i+1 | l_i < l_{i+1},$ $ i\in \mathbb Z, i\geq 0\},$ where $l_0\geq -m_1$. Similarly, by Eq \eqref{eq:odd}, we get $m_0\geq -l_1$.
\end{proof}

From Lemma \ref{lem:Um0},   Eq \eqref{eq:odd} and   Eq \eqref{eq:even}, we need to discuss the following four cases:

\vspace{2mm}(1)  $l_0=-m_1$.

By a direct computation according Eq \eqref{eq:odd} and Eq \eqref{eq:even},
 we have
$$ ~~ m_i=m_1+(i-1)(m_1-m_0), ~ l_1=-m_0, ~l_i=-m_0+(i-1)(m_1-m_0), ~ i\in \mathbb Z, ~ i\geq 1,$$
where
$W_{1}=\{ 2m_i ~ |~  m_i< m_{i+1},  i=0, 1, 2, \cdots\}, $ and  $U_{1}=\{ 2l_i+1 ~ | ~l_i< l_{i+1}, i=0, 1, 2, \cdots\}. $

\vspace{2mm}(2)  $-m_1<l_0<-m_0$.

From $2(m_0+l_0+m_1)\in W_1$, and  $m_0+l_0+m_1<m_1$, we have $m_0+l_0+m_1=m_0$,
this contradicts  $l_0 < -m_1$. Therefore, this case does not exist.

\vspace{2mm}(3) $l_0=-m_0$.

From $f(0)=f(0)f(2m_0)f(2l_0+1)=f(0)f(2m_0)f(-2m_0+1)=-f(0)^2$,

\vspace{2mm}$f(1)=f(1)f(2m_0)f(2l_0+1)=f(1)f(2m_0)f(-2m_0+1)=-f(1)^2$, and

\vspace{2mm}$f(0)+f(1)+1\neq 0$, we have
$f(0)=f(1)=0$ or $f(0)=f(1)=-1$.

\vspace{2mm}$\bullet$ If  $f(0)=f(1)=0$. Then  for all $k, l\in Z$, $k>0$ and $l>0$,

\vspace{2mm}$f(2m_0-2k)f(-2m_0-2l+1)f(0)=f(-2(k+l))=0$,

\vspace{2mm} $f(2m_0-2k)f(-2m_0-2l+1)f(1)=f(-2(k+l)+1)=0$,
\\we obtain   $m_0\geq 1, -m_0=l_0\geq -1$. We  assert that
$$ m_0=1, l_0=-1.$$

In fact, if there is $k_0>1$ such that $f(2k_0)=0$, then $f(-2k_0-2+1)=0$.
Thanks to  Eq \eqref{eq:odd}, we get the contradiction $f(1)f(2k_0)f(-2k_0-2+1)=f(-2+1)=f(2l_0+1)=0.$
Therefore,

$$W_1=\{ 2k, k\in \mathbb Z, k>0\}, ~~ U_1=\{-1,  2k+1, k\in \mathbb Z, k> 0\}.$$

\vspace{2mm}$\bullet\bullet$ If  $f(0)=f(1)=-1$.  For all $l, m, n, s\in Z$, $lmns\neq 0$, if $f(2l+1)=f(2n+1)=f(2m)=f(2s)=-1$, then
$f(2l+2n+1)=f(2m+2s)=f(2l+2m)=f(2l+2m+1)=-1$. We obtain that  $2m_1+2l_0=2m_1-2m_0\in W_1$, $2l_1+2l_0+1=2l_1-2m_0+1\in U_1$.

If $m_0>0$,  by Lemma \ref{lem:Um0},   $m_1-m_0>0$, $l_1-m_0<l_1$.
Then  $m_1=2m_0$, $l_1=m_0$. Inductively, suppose $m_k=(k+1)m_0,$ $l_k=km_0$. Since
$$m_{k-1}=km_0=m_k-m_0<m_{k+1}-m_0<m_{k+1},$$
$$m_{k+1}=(k+2)m_0,~~l_{k-1}=(k-1)m_0=l_k-m_0<l_{k+1}-m_0< l_{k+1}.$$
Then  $l_{k+1}=(k+1)m_0$. Therefore,
$$W_1=\{ 2km_0~ | ~ k\in \mathbb Z, k>0\}, ~~ U_1=\{-2m_0+1,  2km_0+1 ~ | ~ k\in \mathbb Z, k> 0\}.$$

Similarly, if $m_0<0$,  we have $$W_1=\{ 2m_0, -2km_0~ | ~ k\in \mathbb Z, k>0\}, ~~ U_1=\{2km_0+1 ~ | ~ k\in Z, k> 0\}.$$

\vspace{2mm}(4) $l_0>-m_0$.

We can choose  $W_1=\{ 2m_k~ | ~ m_k< m_{k+1}, m_k\in \mathbb Z, k\geq 0\}$, and  $U_1=\{ 2l_k+1~ | ~ l_k<l_{k+1}, k\geq 0\}$.

If there is $m '>m_0$ such that $f(2m')=0$. From $m>m_0$, $-m'<-m_0<l_0$, we have $f(-2m'+1)=0$. By Eq \eqref{eq:odd} and Eq \eqref{eq:even},
$$f(0)f(2m')f(-2m'+1)=(f(0)+1)f(0)=0,$$
$$ f(1)f(2m')f(-2m'+1)=(f(1)+1)f(1)=0.$$
Thanks to $f(0)+f(1)+1\neq 0$,   $f(0)=f(1)=0$, or $f(0)=f(1)=-1$.

$\bullet$ If $f(0)=f(1)=-1$. From $f(2m_0+2l_0)=-1$ and $f(2m_0+2l_0+1)=-1$, we obtain $m_0 > 0$,  $l_0>0$.

 In the case $m_0=l_0$, from $f(k2m_0)=-1$, we have  $l_0=m_0 > 1$. Therefore,

 $\{2km_0~ | ~k\in Z, k >0\}\subseteq W_1$, ~  and  ~~ $\{2km_0+1~ | ~k\in \mathbb Z, k >0\}\subseteq U_1$.

 If there is $0< r <m_0$, $k>0$ such that $f(2m_0k+2r)=0$. From $f(-2r)=f(-2km_0+1)=0$, and Eq \eqref{eq:even},
 we get the  contradiction
 $$
 0=f(2m_0k+2r)f(-2r)f(-2km_0+1)=f(0)=-1.
 $$
Therefore, $f(2m)\neq 0$, for all $m\geq m_0$, that is,
$$\{2km_0~ | ~k\in \mathbb Z, k >0\}= W_1.$$

Similarly we have $f(2m+1)\neq 0$ for all $m\geq l_0$, that is,
$$\{2km_0+1~ | ~k\in \mathbb Z, k >0\}\subseteq U_1.$$

Thanks to Eq \eqref{eq:odd} and Eq \eqref{eq:even}, $f(2m)=f(2m+1)=-1, ~~ \forall m\in\mathbb Z, m\geq m_0.$

 If $l_0\neq m_0$.  From $f(2l_0+2m_0)=f(2m_0+2l_0+1)=-1$, we have

 $\{2km_0+2ln_0~ | ~k, l\in \mathbb Z, k >0, l \geq 0\}\subseteq W_1$, ~  and  ~~ $\{2km_0+2ln_0+1 | k, l\in Z, k \geq 0, l >0\}\subseteq U_1$.

By the similar discussion to the above,
 $W_1=\{2m | m\in \mathbb  Z, m\geq m_0\}$, ~  and  ~~ $U_1=\{2n+1 | n\in \mathbb Z, n\geq l_0\}$, and
 for all $l\in W_1\cup U_1$,  $f(l)=-1.$

 $\bullet$$\bullet$ Now we prove that the case $f(0)=f(1)=0$ does not exist.

 If $f$ satisfies $f(0)=f(1)=0$. From   $l_0>-m_0 > -m'$, $l_0>-m'+1$,  we have $f(2m')=0$, and
 $$f(0)f(2m')f(-2m'+2+1)=(f(0)+1)f(2)=0,$$
  $$f(1)f(2m')f(-2m'+2+1)=(f(0)+1)f(3)=0.$$
Then $f(2)=f(3)=0$. For  $k\in \mathbb Z$, $k > 0$, if $f(2k)=f(2k+1)=0$,  by Eq \eqref{eq:odd} and Eq \eqref{eq:even}, we have
 $$f(0)f(2k)f(2+1)=(f(0)+1)f(2k+2)=f(2k+2)=0,$$
  $$f(1)f(2k)f(2+1)=(f(0)+1)f(2k+2+1)=f(2k+2+1)=0.$$

    Therefore, for all positive $k\in \mathbb Z$, $f(2k)=f(2k+2+1)=0$, this contradicts $|U_1|=\infty$.

Summarizing above discussion, we obtain the following result.

\begin{theorem} Let $R$ be a   homogeneous Rota-Baxter operator on $A_{\omega}$ with $f(0)+f(1)+1\neq 0$, and
$W_1=\{ 2m_i ~|~ i\in Z, i\geq 0,  m_i <m_{i+1}\}$, $U_1=\{ 2l_i+1 ~|~ i\in Z, i\geq 0,  l_i <  l_{i+1}\}$.
Then  the map $f: \mathbb Z\rightarrow \mathbb F$  in Eq \eqref{eq:R0} is one of the following cases:

1) There exist $m_0, m_1\in Z$, $m_0 < m_1$ such that
 for all $~ k\in \mathbb Z, k\geq 0,$
$$f(2m_0)=f(2m_1+ 2k(m_1-m_0))=-1,$$
$$f(-2m_1+1)=f(-2m_0+2k(m_1-m_0)+1)=-1, $$
and $f(m)=0 $ for the remaining $m\in \mathbb Z$.

2) For all $ k\in \mathbb Z, k > 0,$
$$f(2k)=f(-1)=f(2k+1)=-1, $$
and $f(m)=0 $ for the remaining $m\in \mathbb Z$.

3) There is $m_0\in \mathbb Z$, $m_0 > 0$ such that for all $  k\in \mathbb Z, k \geq 0,$
$$f(2km_0)=f(-2m_0+1)=f(2km_0+1)=-1,$$
and $f(m)=0 $ for the remaining $m\in \mathbb Z$.

4)  There is  $m_0\in \mathbb Z$, $m_0 < 0$ such that
 for all $ k\in \mathbb Z, k \leq  0,$
$$f(2m_0)=f(2km_0)=f(2km_0+1)=-1, $$
and $f(m)=0 $ for the remaining $m\in \mathbb Z$.

5) There exist $m_0, l_0\in \mathbb Z$, $l_0 > -m_0$ such that for all $ m, l\in \mathbb Z, m\geq m_0, l\geq l_0,$
$$f(2m)=f(2l+1)=-1, $$
and $f(m)=0 $ for the remaining $m\in \mathbb Z$.

6) $f(0)=f(1)=-1$, and there is $m_0\in \mathbb Z, m_0>1$ such that
 $f(m)=-1$, for all $m\in \mathbb Z$, $m\geq m_0$, and $f(m)=0 $ for the remaining $m\in \mathbb Z$.

 7) $f(0)=f(1)=-1$ and  there exist  $m_0, l_0\in \mathbb Z$ such that $m_0 > 0$, $l_0>0$, $m_0\neq l_0,$
  $f(2m)=f(2n+1)=-1$,  for all $m, n\in \mathbb Z$, $m\geq m_0, n\geq l_0$,
 and $f(m)=0 $ for the remaining $m\in \mathbb Z$.

\label{thm:Rm0}
\end{theorem}

By the similar discussion, we get the following result.

\begin{lemma} Let $R$ be a   homogeneous Rota-Baxter operator on $A_{\omega}$ and $W_1=\{ 2m_i ~|~   m_i > m_{i+1}, i\in \mathbb Z, i\geq 0\}$.
Then $U_1=\{ 2l_i+1 ~| ~ l_i > l_{i+1}, i\in \mathbb Z, i\geq 0 \}$, and $l_0\leq -m_1$, $l_1\leq -m_0$.
\label{lem:Um1}
\end{lemma}

\begin{proof}
 For all $2l+1\in U_1$,  by Eq \eqref{eq:even},  $f(2m_0+2m_1+2l)=-1$. Then $2l+2m_0+2m_1\leq 2m_0$, and $l\leq -m_1$. So we can suppose
 $U_1=\{ 2l_i+1 |  l_i > l_{i+1}, i\in \mathbb Z, i\geq 0 \},$  $l_0\leq -m_1$. Thanks to  Eq \eqref{eq:odd}, $m_0\leq -l_1$.
\end{proof}

\begin{theorem} Let $R$ be a   homogeneous Rota-Baxter operator on $A_{\omega}$ and

$W_1=\{ 2m_i | i\in \mathbb Z, i\geq 0,  m_i > m_{i+1}\}$, $U_1=\{ 2l_i+1 | i\in \mathbb Z, i\geq 0,  l_i >  l_{i+1}\}$.
\\Then  the map $f: \mathbb Z\rightarrow \mathbb F$  in Eq \eqref{eq:R0} is one of the following cases:

1) There is $m_0, m_1\in Z$, $m_0 > m_1$ such that for all $ k\in \mathbb Z, k\leq 0,$
$$f(2m_0)=f(2m_1+ 2k(m_0-m_1))=-1,$$
$$f(-2m_1+1)=f(-2m_0+2k(m_0-m_1)+1)=-1,$$
and $f(m)=0 $ for the remaining $m\in \mathbb Z$.

2) For all $ k\in \mathbb Z, k < 0,$
$$f(2)=f(2k)=f(2k+1)=-1, ~$$
and $f(m)=0 $ for the remaining $m\in \mathbb Z$.

3) There is $m_0\in \mathbb Z$, $m_0 < 0$ such that for all $ ~ k\in \mathbb Z, k \geq 0$,
$$f(2km_0)=f(-2m_0+1)=f(2km_0+1)=-1,$$
and $f(m)=0 $ for the remaining $m\in \mathbb Z$.

4) There is $m_0\in \mathbb Z$, $m_0 >0$ such that for all $~ k\in \mathbb Z, k \leq 0,$
$$f(2m_0)=f(2km_0)=f(2km_0+1)=-1,$$
and $f(m)=0$ for the remaining $m\in \mathbb Z$.

5) There exist $m_0, l_0\in \mathbb Z$, $l_0 <  -m_0$ such that for all $ ~ m, l\in \mathbb Z, m\leq m_0, l\leq l_0,$
$$f(2m)=f(2l+1)=-1,$$
and $f(m)=0$ for the remaining $m\in \mathbb Z$.

6) $f(0)=f(1)=-1$, and there is $m_0\in \mathbb Z, m_0 < -1$ such that
 $f(l)=-1$ for all $l\leq 2m_0+1,$ and $f(m)=0$ for the remaining $m\in \mathbb Z$.

 7) $f(0)=f(1)=-1$ and  there exist $m_0, l_0\in \mathbb Z$, $l_0 < 0$, $m_0 < 0$, $m_0\neq l_0$ such that  for all $m, l\in \mathbb Z, m\leq m_0, l\leq l_0,$  $f(2m)=f(2l+1)=-1$
and $f(m)=0$ for the remaining $m\in \mathbb Z$.

\label{thm:Rm1}
\end{theorem}

\begin{proof}

Apply the arguments used in the proof of Theorem \ref{thm:Rm0}.

\end{proof}

\begin{theorem} If $\inf W_i=\inf U_i=-\infty$ and   $\sup W_i=\sup U_i=+\infty$. Then $R$ is a   homogeneous Rota-Baxter operator on $A_{\omega}$ if and only if  the map
 $f: \mathbb Z\rightarrow \mathbb F$ in Eq \eqref{eq:R0} is
one of the following cases:

1) There is $m_0\in \mathbb Z$, $m_0\neq 0$ such that for all $k\in \mathbb Z, $
$f(2km_0)=f(2m_0k+1)=0,$ and $f(m)=-1$ for the remaining $m\in \mathbb Z$.

2) There is $m_0\in \mathbb Z$, $m_0\neq 0$ such that for all $k\in \mathbb Z, $
$f(2km_0)=f(2m_0k+1)=-1,$ and $f(m)=0$ for the remaining $m\in \mathbb Z$.

\label{thm:R01}
\end{theorem}

\begin{proof}

Let $R$ be a   homogeneous Rota-Baxter operator on $A_{\omega}$.  Suppose

$$W_{2}=\{ 2m_i, 2m_i' | i\in \mathbb Z, i \geq 0\}, ~ U_{2}=\{ 2l+1, 2l'_{i}+1| i\in \mathbb Z, i\geq 0\},$$ where
$$ \cdots <m'_{i+1} < m'_i < \cdots < m'_1 < m'_0 <0 < m_0 < 2m_1<\cdots < m_i < m_{i+1} < \cdots,$$
$$ \cdots <  l'_{i+1} < l'_i < \cdots < l'_1 < l'_0 < 0 <l_0 < l_1 <\cdots < l_i < l_{i+1} < \cdots.$$

 If $f(0)=b\neq 0, -1.$ Thanks to Eq \eqref{eq:odd}, for $l, k\in \mathbb Z$ and $l\neq k$, if $f(2l+1)=f(2k+1)=0, $ then $f(2l+2k+1)=0, $ and $f(2l_0+2l'_0+1)=0$. Since $2l'_0+1 < 2l_0+2l'_0+1<2l_0+1$,  $f(1)=0$. By Eq \eqref{eq:even}, and  $f(2m)=f(2n)=0, m\neq n$,  we have $f(2m+2n)=0.$ From $f(2m_0+2m'_0)=0$, and $2m'_0 <2m_0+2m'_0<2m_0$, we get the contradiction  $f(0)=b=0$.
 Therefore, $f(0)=0$, or $f(0)=-1$.

\vspace{2mm} If $f(0)=0$. By $2l'_0+1< 2l_0+2l'_0+1< 2l_0+1$ and Eq \eqref{eq:odd},
 $f(0)f(2l_0+1)f(2l'_0+1)=f(2l_0+2l'_0+1)=0$. Therefore,  $l_0'=-l_0$ and $f(1)=0$.

Similar discussion, we obtain that for all $ i\in \mathbb Z, i\geq 0,$
$m_i=-m'_{i}, ~~ l_i=-l'_{i}. $
 Therefore,
 for all $2m, 2n\in W_2, 2l+1, 2s+1\in U_2$, we have $2m+2n, 2m+2l\in W_2$ and $2l+2s+1, 2l+2m+1\in U_2$. Since
$0< 2m_1-2m_0=2m_1+2m'_0< 2m_1 $, $2m_1-2m_0=2m_0$, and $m_1=2m_0$.
 Inductively, we have
  $m_i=(i+1)m_0$, ~~ $m'_i=-(i+1)m_0$,~ $l_i=(i+1)l_0$, $l'_i=-(i+1)l_0$, for all $i\in \mathbb Z, i\geq 0$.

 \vspace{2mm} We  affirm $m_0=l_0$.

 In fact, if $m_0\neq l_0$, then $m_0-l_0\neq 0$. From $2m_0-2l_0=2m_0+2l'_0< 2m_0$, $2m_0-2l_0\in W_2$, and  $2l'_0+1<2m_0-2l_0+1\in U_2$, we get the  contradiction $2m_0-2l_0<0$ and $2m_0-2l_0>0$.
 Therefore, $m_0=l_0$. We get case 1).

 By the similar discussion, if $f(0)=-1$, then $f(1)=-1$, and we obtain the  case 2).
\end{proof}

\subsubsection{\bf Homogeneous Rota-Baxter operators with  $f(0)=a\neq 0$ and $f(0)+f(1)+1= 0$ }

In this section we discuss  homogeneous Rota-Baxter operators on $A_{\omega}$ of weight $1$ defined by Eq \eqref{eq:R0}  with  $f(0)=a\neq 0$ and $f(0)+f(1)+1= 0$.

\begin{lemma}
Let $R$ be homogeneous Rota-Baxter operators on $A_{\omega}$. Then the map $f: \mathbb Z\rightarrow \mathbb F$ in Eq \eqref{eq:R0} satisfies  that for all  $l, m, n\in \mathbb Z$,
\\
$1)$ ~~  $af(2l+1)f(2m+1)=\big((a+1)f(2l+1)+(a+1)f(2m+1)$

$+f(2l+1)f(2m+1)+(a+1))\big) f(2l+2m+1), ~~ l\neq m$.
\\
$2)$~~ $-(a+1)f(2m+1)f(2n)$

$=\big(-af(2m+1)-af(2n)
+f(2m+1)f(2n)-a\big)f(2m+2n+1),~~  m\neq 0$.
\\
$3)$ ~~ $af(2l+1)f(2m)$

$=\big((a+1)f(2l+1)+(a+1)f(2m)+f(2l+1)f(2m)+(a+1)\big)f(2l+2m), ~~ m\neq 0
$.
\\
$4)$ ~~ $-(a+1)f(2m)f(2n)$

$=\big(-af(2m)-af(2n)
+f(2m)f(2n)-a\big)f(2m+2n), ~~ m\neq n
$.
\label{lem:fa}
\end{lemma}

\begin{proof}
The result follows from Eq \eqref{eq:odd} and Eq \eqref{eq:even}, directly.

\end{proof}

\begin{theorem}  Let $R$ be a homogeneous Rota-Baxter operators on $A_{\omega}$, then   the map $f: \mathbb Z\rightarrow \mathbb F$  in Eq \eqref{eq:R0}  satisfies equation,  for all  $m\in \mathbb Z$,
\begin{equation}
f(1-m)+f(m)+1=0.
\label{eq:f(m)}
\end{equation}
\label{thm:f(m)f(1-m)}
\end{theorem}

\begin{proof}
By  2) and 3) in Lemma \ref{lem:fa}, for all $m, n\in \mathbb Z, m\neq 0, n\neq 0$,

 \vspace{2mm}$-f(2m+1)f(2n)$
 \\$=f(2m+2n+1)\big(-af(2m+1)-af(2n)+f(2m+1)f(2n)-a\big)
 $
 \\
 $+f(2m+2n)\big((a+1)f(2m+1)+(a+1)f(2n)+f(2m+1)f(2n)+a+1\big).$
\\Then in the case $m=-n$,  we obtain $f(2m+1)+f(-2m)+1=0$.
The result follows.
\end{proof}

\begin{theorem}  Let $R$ be a homogeneous Rota-Baxter operators on $A_{\omega}$,  and $f(2k)\neq 0, f(2l)\neq 0$,
$f(2m+1)\neq 0$, $f(2n+1)\neq 0$, for $k, l, m, n\in \mathbb Z$ and $klmn\neq 0$. Then we have

\vspace{2mm}$1)$~ $f(2k+2l)\neq 0$; \quad $2)$~ $f(2k+2m)\neq 0$;\quad $3)$~ $f(2k+2m+1)\neq 0$; ~

\vspace{2mm} $4)$ ~$f(2m+2n+1)\neq 0$;\quad $5)$ ~$f(2m+2n+2k+1)\neq 0$; \quad $6)$~ $f(2m+2k+2l)\neq 0$;

\vspace{2mm} $7)$ ~ $f(1-2k+2m)\neq 0$, ~~ $m\neq -k$; \quad $8)$ $f(4k)\neq 0$;~~
$9)$ $f(1-2k-2m)+1\neq 0$; ~~

\vspace{2mm} $10)$~  $f(2k-2m)+1\neq 0$; \quad
$11)$ $f(1-4k)+1\neq 0$.

\label{thm:klmn}

\end{theorem}

\begin{proof} The result 1) follows from  4)
in Lemma \ref{lem:fa} of the case ~$m=k$, $n=l$, $k\neq l$.

  The result 2) follows from 2) in Lemma \ref{lem:fa} of the case $m=m$,  $n=k$,  $k\neq 0$.

  The result 3) follows from  3) in Lemma \ref{lem:fa} of the case $l=m$, $m=k$,  $m\neq 0$.

The result 4) follows from  1) in Lemma \ref{lem:fa} of the case $l=m$, $m=n$,  $m\neq n$.

The result 5) and 6) follows from  Eq \eqref{eq:odd} and Eq \eqref{eq:even}, directly.

The result 7) follows from  1) in Lemma \ref{lem:fa} of the case  $l=0$, $2m+1, -2k+1$,  $m\neq -k$.

 The result 8) follows from  3) in Lemma \ref{lem:fa} of the case  $l=k$, $m=k$, $k\neq 0$.

 The result 9),  10)  and 11) follow from  2), 7) and 10) and Eq \eqref{eq:f(m)}, respectively.

\end{proof}

\begin{lemma} If at least one of the subsets $W_i, U_i$, $i=1, 2$ is finite. Then $R$ is not a  homogeneous Rota-Baxter operator on $A_{\omega}$.
\label{lem:finite}
\end{lemma}

\begin{proof}
The result follows from 1), 2),  3) and 4) in Theorem \ref{thm:klmn}, directly.

\end{proof}

\begin{theorem} If  $R$ is a  homogeneous Rota-Baxter operator on $A_{\omega}$, then
$$\inf W_i=\inf U_i=-\infty, ~~\sup W_i=\sup U_i=+\infty,$$
\\and there is
$m_0\in \mathbb Z$, $m_0\neq 0$ such that
\begin{equation}
W_1=\{2m_0k | k  \in \mathbb Z, k\neq 0\}, ~~ U_1=\{2m_0k +1 | k \in \mathbb Z\}.
\label{eq:infm0}
\end{equation}
\label{thm:infinite}
\end{theorem}

\begin{proof}
If there is $m_0\in \mathbb Z$ such that $f(2m_0)\neq 0$, and for all $2m\in W_1$, $2m\geq 2m_0$ (similar discussion for the case $2m\leq 2m_0$ ). By 2) and 8) in Theorem \ref{thm:klmn}, and Lemma \ref{lem:finite},   for all $2m+1\in U_1$, we have $f(2m+2m_0)\neq 0$, and $f(4m_0)\neq 0$. Then $4m_0 > 2m_0$, $2m+2m_0\geq 2m_0$. We obtain that  $m_0>0$, and there is $l_0\in \mathbb Z$, $l_0> 0$  such that for all $2l+1\in U_1$, $2l+1\geq 2l_0+1.$
From 7) in Theorem \ref{thm:klmn}, $f(1+2l_0-2m_0)\neq 0$. We get the contradiction  $2l_0+1\leq 1+ 2l_0-2m_0 < 1+ 2l_0$.
Therefore, $\inf W_i=\inf U_i=-\infty, ~~\sup W_i=\sup U_i=+\infty.$

Then we can suppose $$W_{1}=\{ 2m_i, 2m_i' | i\in \mathbb Z, i \geq 0\}, ~ U_{1}=\{ 2l+1, 2l'_{i}+1| i\in \mathbb Z, i\geq 0\},$$
where
$$ \cdots < m'_{i+1} < m'_i < \cdots < m'_1 < m'_0 <0 < m_0 < m_1<\cdots < m_i < m_{i+1} < \cdots,$$
$$ \cdots <  l'_{i+1} < l'_i < \cdots < l'_1 < l'_0 < 0 <l_0 < l_1 <\cdots < l_i < l_{i+1} < \cdots,$$

Thanks to Theorem \ref{thm:klmn},  $2m_0+2m_0'\in W_1$, $m_0'< m_0+m_0'< m_0$. We obtain $m_0'=-m_0$.

Since $0<2m_1+m_0'=2m_1-2m_0< 2m_1$, $m_1=2m_0$. Inductively, we get
$$m_i=(i+1)m_0, ~~ m_i'=-(i+1)m_0, ~~ i\in \mathbb Z, i\geq 0.$$

Similar discussion, for all $i\in \mathbb Z, i\geq 0$, $l_i=(i+1)l_0$ and $l_i'=-(i+1)l_0$.

From  2) and 3) in Theorem \ref{thm:klmn}, there exist positive $s, t\in \mathbb Z$ such that
$$2l_0+2m_0=2s m_0=2tl_0. $$
Then $l_0=(s-1)m_0, m_0=(t-1)l_0$. It follows  $l_0=m_0$.  The proof is complete.
\end{proof}

The subset  $T_{m_0}=W_1\cup U_1$ is called the {\bf $m_0$-supporter of the
  homogeneous Rota-Baxter operator $R$.} Then  for all $m\in \mathbb Z, m\neq 0, 1$, $f(m)\neq 0$ if and only if $m\in T_{m_0}$.

\begin{coro}
Let ~$R$ be a  homogeneous Rota-Baxter operator. Then the map $f:\mathbb Z\rightarrow \mathbb F$ in Eq \eqref{eq:R0} satisfies that for $k\in \mathbb Z$, if
$f(2m_0k)\neq 0$, then $f(2km_0)\neq -1, $  $f(1+2km_0)\neq 0, -1, $ and
\begin{equation}
\frac{1}{f(2m_{0}k)}+\frac{1}{f(-2m_{0}k)}+\frac{1}{f(2m_{0}k)f(-2m_{0}k)}=\frac{1+2a}{a^{2}},
\label{eq:00}
\end{equation}
where $f(0)=a\neq 0.$
\label{cor:1}
\end{coro}

\begin{proof} From 9) and 10) in Theorem \ref{thm:klmn}, if $f(2m_0k)\neq 0$, then $f(2km_0)\neq -1, $  $f(1+2km_0)\neq 0, -1.$
Thanks to  4) in Lemma \ref{lem:fa}, for
$m=-n=2m_{0}k$, ~$k\in \mathbb Z$, $k\neq 0$, we have

$-(1+a)f(2m_{0}k)f(-2m_{0}k)$
$=-a^{2}f(2m_{0}k)-a^{2}f(-2m_{0}k)$ $+af(2m_{0}k)f(-2m_{0}k)-a^{2}.$
\\Since $\frac{1}{f(0)}=\frac{1}{a}$, we obtain Eq \eqref{eq:00}.
\end{proof}

\begin{coro}
Let ~$R$ be a  homogeneous Rota-Baxter operator with $m_0$-supporter $T_{m_0}$. Then the map $f:\mathbb Z\rightarrow \mathbb F$ in Eq \eqref{eq:R0} satisfies that for all $k_1, k_2, k_3\in \mathbb Z$, $k_2\neq k_3$,
\begin{equation}
\frac{1}{f(2m_{0}k_{1})}+\frac{1}{f(2m_{0}k_{1})f(2m_{0}(-k_{1}+k_{2}+k_{3}))}+\frac{1}{f(2m_{0}(-k_{1}+k_{2}+k_{3}))}
\label{eq:01}
\end{equation}
$$=\frac{1}{f(2m_{0}k_{2})}
+\frac{1}{f(2m_{0}k_{3})f(2m_{0}k_{2})}+\frac{1}{f(2m_{0}k_{3})}.
$$
\label{cor:2}
\end{coro}

\begin{proof}

The result follows from 9) and 10) in Theorem \ref{thm:klmn}, and  Theorem \ref{thm:f(m)f(1-m)}.
\end{proof}

\begin{theorem}
Let $R$ be a  homogeneous Rota-Baxter operator on $A_{\omega}$. Then there is $m_0\in \mathbb Z,~~  m_0\neq 0$ such that the map $f:\mathbb Z\rightarrow \mathbb F$ in Eq \eqref{eq:R0} satisfies
one of the two cases:

\vspace{2mm}$1)$  For all $ k\in \mathbb Z, $ $f(2m_0k)=f(0)=a, ~~f(2m_0k+1)=f(1)=-1-a, $
 and $f(m)=0$ for the remaining $m\in \mathbb Z$.

\vspace{2mm}$2)$ ~~ If there is $k_0\in \mathbb Z, k_0\neq 0$ such that $f(2m_0k_0)\neq a$, then  $a\neq -1, -\frac{1}{2}$, and for all $k\in \mathbb Z,$

$$f(4m_0k)=a, ~~ f(4m_0k+1)=-1-a,$$

\begin{equation}   f(4m_0k+2)=\frac{-a}{1+2a}, f(4m_0k+3)=-\frac{1+a}{1+2a},
\label{eq:f(2m0k)a}
\end{equation}
 and $f(m)=0$ for the remaining $m\in \mathbb Z$.

\label{thm:f(0)a}
\end{theorem}

\begin{proof}   If  for all $k\in \mathbb Z$, $f(2m_0k)=a$, then we get the case 1).

Now we prove the case 2).

By Theorem \ref{thm:infinite}, if $R$ is a  homogeneous Rota-Baxter operator, then  the map $f:\mathbb Z\rightarrow \mathbb F$ in Eq \eqref{eq:R0} satisfies that there is $m_0\in \mathbb Z$, $m_0\neq 0$ such that for all $l, k\in \mathbb Z,$ $f(l)\neq 0$ if and only if $l=2m_0k$ or
$l=2m_0k+1$.
From Theorem \ref{thm:infinite}, and 9) and 10) in Theorem \ref{thm:klmn}, for all $k\in \mathbb Z$,
$$f(2km_0)\neq -1, ~~ f(1+2km_0)\neq  -1.$$
Then for $k\neq 0$, let  $k_1=k_3=k, k_2=-k_1$ in Eq \eqref{eq:01}, we obtain

 \begin{equation}
 f(2m_0k)=f(-2m_0k).
 \label{eq:-kk}
 \end{equation}

  Thanks to Eq \eqref{eq:f(m)}, for all $k\in \mathbb Z$, $k\neq 0$,
 \begin{equation}
 f(1+2m_0k)=f(1-2m_0k)=-1-f(2m_0k).
  \label{eq:-kodd}
 \end{equation}

 From Eq \eqref{eq:even}, and Eq \eqref{eq:f(m)},  for all nonzero $l, k\in \mathbb Z$, and $l\neq k$,  we have

\begin{equation}
(f(2m_0k)-f(2m_0l))(f(2m_0k)+2f(2m_0k)f(2m_0l)+f(2m_0l))=0,
\label{eq:2m0kl}
\end{equation}

\begin{equation}
(f(2m_0k)-a)(f(2m_0k)+2af(2m_0k)+a)=0.
\label{eq:2m0ka}
\end{equation}

Follows from  Eq \eqref{eq:2m0kl}, Eq \eqref{eq:2m0ka}, Eq \eqref{eq:00}, if $f(2m_0l)\neq a$, then $~ a\neq -1, \frac{-1}{2},$ and

$$f(2m_0l)=f(-2m_0l)=\frac{-a}{1+2a}, f(2m_0l+1)=f(-2m_0l+1)=-\frac{1+a}{1+2a}.$$

If there exist $n_0, k_0\in \mathbb Z, k_0\neq 0,  n_0\neq 0$ such that $f(2m_0k_0)\neq a$ and $f(2m_0n_0)=a,$ then $k_0\neq n_0$ and $k_0\neq -n_0$.  Thanks to Eq \eqref{eq:even},
  $f(2m_0(n_0+k_0))\neq a$.

  Similar discussion, for $n_1, k_1\in \mathbb Z$, $k_1\neq k_0$, $n_1\neq n_0$, if $f(2m_0n_1)=a$, $f(2m_0k_1)\neq a$,  by Eq \eqref{eq:01}, $f(2m_0(k_0+k_1))=f(2m_0(n_0+n_1))=a$.
   Without loss of generality, we can suppose that $m_0 >0$. And let $k_0, n_0\in \mathbb Z$ be the least positive satisfying that $f(2m_0k_0)\neq a$ and $f(2m_0n_0)=a$.
   By the above discussion and  Eq \eqref{eq:-kk}, $f(2m_0(k_0-n_0))\neq a$. Since  $k_0-n_0<k_0$, $k_0<n_0$,
   and  $k_0=1$. If $n_0>2$, then $f(2m_02)\neq a$, and
   $f(2m_0(1+2))=a$, we obtain  $n_0=3$. From  $f(2m_0(2+3))\neq a$, and $f(2m_02)\neq a$, we have $f(2m_0(2+5))=a.$ From $f(2m_0(1+3))\neq a$, $f(2m_03)=a$,
  we get the contradiction $f(2m_0(3+4))\neq a$.

  Therefore,  we have $n_0=2,$ and that $f(2m_0k)=a$ if and only if $k=2l$, and  $f(2m_0k)\neq a$ if and only if $k=2l+2$, where $l\in \mathbb Z$.  Eq \eqref{eq:f(2m0k)a} follows.
\end{proof}

\subsubsection{\bf Homogeneous Rota-Baxter operators with  $f(0)= 0$ and $f(1)=-1$ }

In this section we discuss homogeneous Rota-Baxter operators $R$ on $A_{\omega}$ with $R(L_0)=0$, and $R(L_1)=-L_1.$ Thanks to Eq \eqref{eq:R0}, the map $f:\mathbb Z\rightarrow \mathbb F$ satisfies  $f(0)= 0$ and $f(1)=-1$.

\begin{lemma}Let $R$ be a homogeneous Rota-Baxter operator on $A_{\omega}$ with $R(L_0)=0$, and $R(L_1)=-L_1.$ Then the map $f:\mathbb Z\rightarrow \mathbb F$ in Eq \eqref{eq:R0} satisfies following conditions, for all  $l, m, n\in \mathbb Z$,

\vspace{2mm}$1)$ ~~  $(f(2l+1)+1)(f(2m+1)+1)f(2l+2m+1)=0,  ~~ l\neq m$.

\vspace{2mm}$2)$ ~~ $f(2m+1)f(2n)(1+f(2m+2n+1))=0, ~~~  m\neq 0$.

\vspace{2mm}$3)$ ~~  $(f(2l+1)+1)(f(2m)+1)f(2l+2m)=0,  ~~ m\neq 0$.

\vspace{2mm}$4)$ ~~  $f(2m)f(2n)(1+f(2m+2n))=0,  ~~ m\neq n$.

\label{lem:f00}
\end{lemma}

\begin{proof}
The result follows from Eq \eqref{eq:odd}, Eq \eqref{eq:even},   $f(0)=0$ and $f(1)=-1$, directly.
\end{proof}

\begin{coro}  Let $R$ be a homogeneous Rota-Baxter operator on $A_{\omega}$ with $R(L_0)=0$, and $R(L_1)=-L_1.$  Then the map $f:\mathbb Z\rightarrow \mathbb F$ in Eq \eqref{eq:R0} satisfies that for all $k, l, m, n\in \mathbb Z$ , $klmn\neq 0$,

\vspace{2mm}$1)$ if $f(2k)\neq 0, f(2l)\neq 0$, $k\neq l$, $k\neq -l$, then $f(2k+2l)=-1$.

\vspace{2mm}$2)$ If $f(2k)\neq 0, f(2m+1)\neq 0$, $m\neq 0$, then ~$f(2k+2m+1)=-1$.

\vspace{2mm}$3)$ If  $f(2k)= 0, f(2n+1)=0$, $k\neq 0$, then ~$f(2k+2n)=0$.

\vspace{2mm}$4)$ If $f(2m+1)= 0, f(2n+1)=0$,  $m\neq n$, $m\neq -n$, then ~$f(2m+2n+1)=0$.

\vspace{2mm}$5)$ If $k\neq 0$, $f(2k)f(-2k)= 0$.

\vspace{2mm}$6)$ For all $m\in \mathbb Z$, ~ $(f(2m+1)+1)(f(-2m+1)+1)=0$.

\vspace{2mm}$7)$ ~ $|W_2|=|U_1|=\infty.$

\label{cor:01lkmn}
\end{coro}

\begin{proof}
The result follows from Lemma \ref{lem:f00}, directly.
\end{proof}

\begin{theorem} If $|W_1| < \infty$, then
$R$ is a  homogeneous Rota-Baxter operator on $A_{\omega}$ if and only if the map $f:\mathbb Z\rightarrow \mathbb F$ in Eq \eqref{eq:R0} satisfies  one of the following conditions:

\vspace{2mm}$1)$ ~~  $|W_1|=|U_2|=0$, and  for all $m\in \mathbb Z, f(2m)=0$, $f(2m+1)=-1$.

\vspace{2mm}$2)$ ~~  $|W_1|=|U_2|=0$,   and there is nonzero $n_0\in \mathbb Z$ such that $f(2n_0+1)\neq 0, -1$,  and for all $m, n\in \mathbb Z, f(2m)=0$, $f(2n+1)=-1$, $n\neq n_0$.

\vspace{2mm}$3)$ ~~  $|W_1|=0, $ $|U_2|=1$, and there is nonzero $n_0\in \mathbb Z$ such that $f(2n_0+1)=0$, and for all $m, n\in \mathbb Z,$ $f(2m)=0$, $f(2n+1)=-1$, $n\neq n_0$.

\vspace{2mm}$4)$ ~~$|W_1|=1$, $|U_2|=0$,  and there is nonzero $m_0\in \mathbb Z$ such that $f(2m_0)\neq 0$ and for all $m, n\in \mathbb Z$, $f(2m)=0$, $f(2n+1)=-1$,  $m\neq m_0$.

\label{thm:f(0)a1}
\end{theorem}

\begin{proof}
The discussion is completely similar to Theorem \ref{thm:fin}.
\end{proof}

From Theorem \ref{thm:f(0)a1}, Let $R$ be a  homogeneous Rota-Baxter operator   with $R(L_0)=0$ and $R(L_1)=-L_1$. Then $|W_1|\neq 0$ and $|U_2|\neq 0$ if and only if $|W_1|=|U_2|=\infty$.
So in the following we discuss the case  $|W_1|=|U_2|=\infty$.

\begin{theorem} Let $|W_1|=\infty$, then $R$ is a  homogeneous Rota-Baxter operator  with $R(L_0)=0$ and $R(L_1)=-L_1$ if and only if the map $f:\mathbb Z\rightarrow \mathbb F$ in Eq \eqref{eq:R0} is  one of the following cases:

\vspace{2mm}$(1)$ There is $m_0, n_0\in \mathbb Z$, $m_0>0$,  $n_0 <0$  such that for all $m, n\in \mathbb Z,$ $f(2m)=0$ if and only if $m < m_0$,
and $f(2n+1)=0$ if and only if $n\leq n_0$. And $f$ is one of the following seven cases:

\vspace{2mm} $1)$ $f(2n+1)=f(2m)=-1$, for all $n > n_0$, $m\geq m_0$,  and $f(m)=0$ for the remaining $m\in \mathbb Z$.

\vspace{2mm} $2)$ There exist $ c, d\in \mathbb F, $ $cd\neq 0$ and $ c\neq -1, ~or~ d\neq -1 $ such that for all $ m, n\in \mathbb Z,$ $m\geq m_0$
$$f(2m)=-1, f(2n+1)=-1, f(-1)=c, f(-3)=d,  n\geq 0,  $$
 and $f(m)=0$ for the remaining $m\in \mathbb Z$
(in this case $n_0=-3$).

\vspace{2mm} $3)$  There exist $ c'\in \mathbb F, $ $c'\neq 0$ and $ c'\neq -1,  $ for all $ m, n\in \mathbb Z,$ $m\geq m_0,  n \geq 0, n\neq 1,$
$$f(2m)=-1, f(2n+1)=f(-1)=f(-3)=-1,  f(3)=c',  $$
 and $f(m)=0$ for the remaining $m\in \mathbb Z$ (in this case $n_0=-3$).

 \vspace{2mm}  $4)$ There is $ g\in \mathbb F, $ $g\neq 0, -1 $ such that  for all $ m, n\in \mathbb Z,$ $ m\geq m_0,  n \geq  0,$
$$f(2m)=-1, f(2n+1)=-1, f(-1)=g,  $$
 and $f(m)=0$ for the remaining $m\in \mathbb Z$ (in this case $n_0=-2$).

\vspace{2mm} $5)$ There is $m_1\in \mathbb Z$, $m_1\geq m_0$, $h\in \mathbb F$, $h\neq 0, -1$ such that for all $ m, n\in \mathbb Z,$ $ m\geq m_0,$ $ n>n_0,$
$$f(2m_1)=h, f(2m)=-1, f(2n+1)=-1,  m\neq m_1, $$
 and $f(m)=0$ for the remaining $m\in \mathbb Z$.

$6)$ There is $m_1, n_1\in \mathbb Z$, $m_1\geq m_0$, $n_1 > n_0$, $h, h'\in \mathbb F$, $h, h'\neq -1$ and $hh'\neq 0$ such that for all $ m, n\in \mathbb Z,$ $m\geq m_0$, $n>n_0,$
$$  f(2m_1)=h, f(2n_1+1)=h', f(2m)=-1, f(2n+1)=-1,  m\neq m_1,  n\neq n_1, $$
 and $f(m)=0$ for the remaining $m\in \mathbb Z$.

$7)$  There is $m_1, m_2\in \mathbb Z$, $m_1, m_2\geq m_0$, $m_1\neq m_2$, and  $g, r\in \mathbb F$, $g, r\neq -1$, $gr\neq 0$, such that for all  $m, n\in \mathbb Z$,
$m\geq m_0$, $n > n_0$,
$$f(2m_1)=g, f(2m_2)=r, f(2n+1)=f(2m)=-1 ,   m\neq m_1, m_2,$$
 and $f(m)=0$ for the remaining $m\in \mathbb Z$.

\vspace{4mm}$(2)$ There exist  $m_0, n_0\in \mathbb Z$, $m_0< 0$ and  $n_0 >0$ such that for all $m, n\in \mathbb Z,$ $f(2m)=0$ if and only if $m > m_0$,
and $f(2n+1)=0$ if and only if $n\geq n_0$. And $f$ is  one of the seven cases:

\vspace{2mm} $1)'$ $f(2n+1)=-1$, and $f(2m)=-1$ for all $n < n_0$, $m\leq m_0$, and $f(m)=0$ for the remaining $m\in \mathbb Z$.

\vspace{2mm} $2)'$ There is $ c\in \mathbb F, $ $c\neq 0$ and $ c\neq -1 $ such that for all $ m, n\in \mathbb Z,$ $m\leq m_0$,
$$f(2m)=-1, f(2n+1)=-1, f(3)=c,  n \leq 0, $$
and $f(m)=0$ for the remaining $m\in \mathbb Z$.

 $3)'$  There exist $ c', d'\in \mathbb F, $ $c'd'\neq 0$ and $ c'\neq -1, ~or~ d'\neq -1 $ such that for all $ m, n\in \mathbb Z,$ $m\leq m_0,  n < -2,$
$$f(2m)=-1, f(2n+1)=f(1)=f(3)=-1, f(-1)=c', f(-3)=d',  $$
and $f(m)=0$ for the remaining $m\in \mathbb Z$.

  $4)'$ There is $ g\in \mathbb F, $ $g\neq 0, -1 $ such that for all $ m, n\in \mathbb Z,$ $ m\geq m_0,  n \leq  0,$
$$f(2m)=-1, f(2n+1)=-1, f(-1)=g,  n\neq -1, $$
and $f(m)=0$ for the remaining $m\in \mathbb Z$.

$5)'$ There exist $m_1\in \mathbb Z$, $m_1\leq m_0$, $h\in \mathbb F$, $h\neq 0, -1$ such that for all $ m, n\in \mathbb Z,$ $ m\geq m_0,$ $ n<n_0,$
$$f(2m_1)=h, f(2m)=-1, f(2n+1)=-1,  m\neq m_1, $$
and $f(m)=0$ for the remaining $m\in \mathbb Z$.

$6)'$ There exist $m_1, n_1\in \mathbb Z$, $m_1\leq m_0$, $n_1 < n_0$, $h, h'\in \mathbb F$, $h, h'\neq -1$ and $hh'\neq 0$ such that for all $ m, n\in \mathbb Z,$ $m\leq m_0$, $n<n_0,$
$$  f(2m_1)=h, f(2n_1+1)=h', f(2m)=-1, f(2n+1)=-1,  m\neq m_1,  n\neq n_1, $$
and $f(m)=0$ for the remaining $m\in \mathbb Z$.

$7)'$  There exist $m_1, m_2\in \mathbb Z$, $m_1, m_2\leq m_0$, $m_1\neq m_2$, and  $g, r\in \mathbb F$, $g, r\neq -1$, $gr\neq 0$ such that for all  $m, n\in \mathbb Z$,
$m\leq m_0$, $n < n_0$,
$$f(2m_1)=g, f(2m_2)=r, f(2n+1)=f(2m)=-1 ,   m\neq m_1, m_2,$$
and $f(m)=0$ for the remaining $m\in \mathbb Z$.

\label{thm:f(0)a3}
\end{theorem}

\begin{proof} $ (i).$~~  We first discuss $W_i$ and $U_i$, for $i=1, 2$.

Since $|W_1|=\infty$, without loss of generality, we can suppose that there is $m\in \mathbb Z$, $f(2m)\neq 0$ and $m>0$.

Then there is $2m_0\in \mathbb Z$ such that $2m_0$ is the least positive which is contained in $W_1$. We will prove that $W_1=\{ 2m | m\in \mathbb Z, m\geq m_0\}$ and $U_2=\{ 2n+1 | n\in \mathbb Z, n\leq n_0\}$.

If for all $n<0$, $f(2n+1)\neq 0$. By Corollary \ref{cor:01lkmn},
$f(2n+k2m_0+1)=-1,$ for all $k\in \mathbb Z$, $k>0$. We get the contradiction $|U_2|=0$.
Therefore,  there is  the largest negative $2n_0+1\in \mathbb Z$ such that $f(2n_0+1)=0$, that is,  $2n_0+1\in U_2$, $n_0<0$.

\vspace{2mm}First, if there is $m\in \mathbb Z$, $m <0$ such that $2m\in W_1$. Let $2m_0'\in \mathbb Z$ be the largest negative which is contained in $W_1$.  By Corollary \ref{cor:01lkmn},
 $2m_0'+2m_0\in W_1$. Since $2m_0'< 2m_0'+2m_0< 2m_0$, $m_0'=-m_0$. This contradicts to 5) in Corollary \ref{cor:01lkmn}.  Therefore, for all $2m\in W_1$, $m\geq m_0$.

If  there is $m>n_0$ such that $2m+1\notin U_1$, then $f(2m+1)=0$. Let $2m'\in U_2$ be the least one which satisfies $m'>n_0$.
From $f(2m'+2n_0+1)= 0$ and $n_0< 0$, we get $ 2m'+2n_0< 2m'$. Therefore, $2m'+2n_0 < 2n_0$, and $m' < 0$. By the nature of $n_0$, we get the contradiction  $n_0 > m'$.
Therefore, for all $2n+1\in U_1, $ $n>n_0$.

Summarizing above discussion, we have that  for all $m, n\in \mathbb Z$,  $m<m_0$ and $n\leq n_0$, $f(2n+1)=0$ and $f(2m)=0$. Thanks to Corollary \ref{cor:01lkmn}, $f(2n+1)=-1$ for $n>-n_0$ and $f(2n+1)\neq 0$ for $n_0< n< 0$, $f(2m)=0$ for all $0< m< m_0$.

If there is $n\in \mathbb Z$  such that $0< n< -n_0$ and $f(2n+1)=0$. Let $n''\in \mathbb Z$  be the least one satisfying $f(2n+1)=0$, $0< n<-n_0$. Then $f(2n_0+2n''+1)=0$. We get the contradiction  $2n_0+1<2n_0+2n''+1< 2n''+1$.   Therefore, for all
$n\in \mathbb Z$, $0<n<-n_0$, $f(2n+1)\neq 0$.

If there is $m\in \mathbb Z$ such that $-m_0< m< 0$ and $f(2m)=0$.  Let $m''\in \mathbb Z$, $-m_0< m'' <0$ be the largest one satisfying $f(2m'')\neq 0$. Then  $f(2m_0+2m'')\neq 0$. But $2m''< 2m_0+2<m''< 2m_0$.  We get the contradiction.
Therefore,  there exist $m_0, l_0\in \mathbb Z, m_0 > 0$ and $n_0 < 0,$ such that

$W_1=\{ 2m | m\in \mathbb Z, m\geq m_0\}$, $W_2=\{ 2m | m\in \mathbb Z, m < m_0\}$,

$U_1=\{ 2n+1 | n\in \mathbb Z, n> n_0\}$, $U_2=\{ 2n+1 | n\in \mathbb Z, n\leq n_0\}$.

\vspace{2mm}Similar discussion, if there is $m\in \mathbb Z$, $m< 0$ such that $f(2m)\neq 0$, then  there exist $m_0, l_0\in \mathbb Z, m_0 <  0$ and $n_0 > 0,$ such that

$W_1=\{ 2m | m\in \mathbb Z, m\leq m_0\}$, $W_2=\{ 2m | m\in \mathbb Z, m > m_0\}$,

$U_1=\{ 2n+1 | n\in \mathbb Z, n< n_0\}$, $U_2=\{ 2n+1 | n\in \mathbb Z, n\geq n_0\}$.

\vspace{3mm}$(ii).$~~  Now  we discuss the characteristic of the map $f$.

From above discussion, we need to discuss the case that  $f(2m)\neq 0$ if and only if $m\geq m_0 >0$, and  $f(2n+1)\neq 0$ if and only if $n>n_0$, $n_0 < 0$.

From  Corollary \ref{cor:01lkmn}, Eq \eqref{eq:odd} and  Eq \eqref{eq:even},  for all positive $l, k, s\in \mathbb Z$,
$l\neq k$,

\begin{equation}
(f(2m_0+2s)+1)(f(2n_0+2k+1)+1)(f(2n_0+2l+1)+1)=0,
\label{eq:skn0}
\end{equation}

\begin{equation}
(f(2n_0+2s+1)+1)(f(2m_0+2k)+1)(f(2m_0+2l)+1)=0.
\label{eq:skn01}
\end{equation}

Then we have

$\bullet$ the case $f(2m)=-1$ for all $m\in \mathbb Z$, $m\geq m_0$.

\vspace{2mm}  If $f(2n+1)=-1$, for all $n> n_0$,  we obtain case $1)$.

 \vspace{2mm} If there is $n_1\in \mathbb Z$, $n_1>n_0$ and $f(2n_1+1)\neq -1$. By Corollary \ref{cor:01lkmn}, Eq \eqref{eq:odd} and  Eq \eqref{eq:even}, we have  $l_0\geq -3$, and
\\$f(2n+1)\begin{cases}
=-1,&  \text{if} ~ n\geq  -n_0, \\
\neq 0,&  \text{if} ~ n_0< n<0,\\
=-1,&  \text{if} ~ 0< n< -n_0;
\end{cases}$
or
$f(2n+1)\begin{cases}
=-1,&  \text{if} ~ n_0< n<0, \\
\neq 0,&  \text{if} ~ 0< n< -n_0.\\
\end{cases}$
\\
Therefore, if $l_0=-3$, we get 2) and 3). If $l_0=-2$, we obtain case 4).

$\bullet$$\bullet$ If there is unique $m_1\in \mathbb Z$, $m_1\geq m_0$ such that $f(2m_1)\neq 0,-1$. Then by Eq \eqref{eq:skn0}, $f(2n+1)=-1$ for all $n\in \mathbb Z$, $n> n_0$; or there is  unique $n_1\in \mathbb Z$, $n_1>n_0$ such that $f(2n_1+1)\neq 0, -1$, and $f(2n+1)=-1$ for all $n\in \mathbb Z,$ $n>n_0$ and $n\neq n_1$. We obtain
$f(2n+1)=-1$ for $n>n_0$. This is case 5). If there is $n_1 > n_0$ such that $f(2n_1+1)\neq -1$, we obtain case 6).

$\bullet$$\bullet$$\bullet$ If the subset  $S=\{m_k | m_k\in \mathbb Z, m_k\geq m_0, f(2m_k)\neq0, -1, k\in \mathbb Z\}$ is non-empty. By Eq \eqref{eq:skn0} and Eq \eqref{eq:skn01},
$S=\{m_1\}$,  or $S=\{m_1, m_2\}$. we obtain 5) and 6), and 7).

\vspace{5mm} The case (2) ( $m_0<0$ and $n_0>0$) follows from the similar discussion.
 \end{proof}

\noindent
{\bf Acknowledgements. } The first author was supported in part by the Natural
Science Foundation (11371245) and the Natural
Science Foundation of Hebei Province (A2014201006).

\bibliography{}

\end{document}